\documentclass[orivec]{llncs}

\usepackage{amsmath,amssymb}

\usepackage{syntax}
\usepackage{paralist}

\usepackage{url}

\usepackage{graphicx}
\usepackage{caption,subcaption}
\graphicspath{{figures/}}

\usepackage{tikz}
\usetikzlibrary{positioning,calc}

\usepackage{algorithm}

\newcommand{\Fig}[1]{Fig.~\ref{#1}}

\usepackage{algpseudocode}
\usepackage{stmaryrd}

\newcommand{\nil}{\ensuremath{\mathsf{nil}}}
\newcommand{\none}{\ensuremath{\mathsf{none}}}
\newcommand{\true}{\ensuremath{\mathsf{true}}}
\newcommand{\false}{\ensuremath{\mathsf{false}}}

\newcommand{\solve}{\textsc{Solve}}
\newcommand{\assertWellFormed}{\textsc{AssertWellFormed}}

\newcommand{\solveTupleTuple}{\textsc{SolveTupleTuple}}

\newcommand{\solveNilRecord}{\textsc{SolveNilRecord}}
\newcommand{\solveRecordRecord}{\textsc{SolveRecordRecord}}
\newcommand{\solveChoiceChoice}{\textsc{SolveChoiceChoice}}
\newcommand{\solveSwitch}{\textsc{SolveSwitch}}

\newcommand{\cspkpn}{\textsc{CSP-KPN}}

\newcommand{\upvar}[1]{{\uparrow}#1}
\newcommand{\downvar}[1]{{\downarrow}#1}
\newcommand{\var}[1]{{\talloblong}#1}
\newcommand{\utvar}{\ensuremath{\mathcal{V}^{\uparrow}}}
\newcommand{\dtvar}{\ensuremath{\mathcal{V}^{\downarrow}}}
\newcommand{\bvar}{\ensuremath{\mathcal{F}}}

\newcommand{\utvars}{\ensuremath{\mathcal{V}^{\uparrow}}}
\newcommand{\dtvars}{\ensuremath{\mathcal{V}^{\downarrow}}}
\newcommand{\bvars}{\ensuremath{\mathcal{{F}}}}
\newcommand{\rel}{\sqsubseteq}
\newcommand{\constr}[2]{\ensuremath{#1 \rel #2}}

\newcommand{\ssat}{\ensuremath{\mathsf{B}}}

\newcommand{\ver}{\ensuremath{\mathsf{V}}}
\newcommand{\ed}{\ensuremath{\mathsf{E}}}
\newcommand{\con}[1]{\ensuremath{\mathcal{C} (#1)}}

\newcommand{\lab}{\ensuremath{\mathsf{L}}}
\newcommand{\term}{\ensuremath{\mathsf{Term}}}
\newcommand{\g}{\ensuremath{\mathsf{G}}}

\newcommand{\unsat}{\ensuremath{\mathsf{Unsat}}}

\newcommand{\el}[3]{%
    \ifx$#1$
        #2{:}~#3
    \else%
        {#1}\ifx$#2$ \else(#2)\fi{:}~#3
    \fi
}

\newcounter{listcount}\newcounter{totalcount}%
\usepackage{etoolbox}% http://ctan.org/pkg/etoolbox
\newcommand{\tuple}[1]{%
  \setcounter{totalcount}{0}% Reset total count
  \renewcommand*{\do}[1]{\stepcounter{totalcount}}% Reconfigure count
  \docsvlist{#1}% Count number of items
  \setcounter{listcount}{0}% Reset current item count
  \renewcommand*{\do}[1]{% Reconfigure item \do
    \stepcounter{listcount}% Next item
    ##1\ifnum\value{listcount}<\value{totalcount}\,\fi% Print item
  }
  \ensuremath{\left(\docsvlist{#1}\right)}% Process list
}
\newcommand{\tlist}[2]{%
  \setcounter{totalcount}{0}% Reset total count
  \renewcommand*{\do}[1]{\stepcounter{totalcount}}% Reconfigure count
  \docsvlist{#1}% Count number of items
  \setcounter{listcount}{0}% Reset current item count
  \renewcommand*{\do}[1]{% Reconfigure item \do
    \stepcounter{listcount}% Next item
    ##1\ifnum\value{listcount}<\value{totalcount},\else\fi% Print item
  }
  \ensuremath{\left[\docsvlist{#1}\ifx$#2$\else\,| #2\fi\right]}% Process list
}
\newcommand{\record}[2]{%
  \setcounter{totalcount}{0}% Reset total count
  \renewcommand*{\do}[1]{\stepcounter{totalcount}}% Reconfigure count
  \docsvlist{#1}% Count number of items
  \setcounter{listcount}{0}% Reset current item count
  \renewcommand*{\do}[1]{% Reconfigure item \do
    \stepcounter{listcount}% Next item
    ##1\ifnum\value{listcount}<\value{totalcount}, \else\fi% Print item
  }
  \ensuremath{\left\{\docsvlist{#1}\ifx$#2$\else\,| #2\fi\right\}}% Process list
}
\newcommand{\choice}[2]{%
  \setcounter{totalcount}{0}% Reset total count
  \renewcommand*{\do}[1]{\stepcounter{totalcount}}% Reconfigure count
  \docsvlist{#1}% Count number of items
  \setcounter{listcount}{0}% Reset current item count
  \renewcommand*{\do}[1]{% Reconfigure item \do
    \stepcounter{listcount}% Next item
    ##1\ifnum\value{listcount}<\value{totalcount},\else\fi% Print item
  }
  \ensuremath{{(}{:}\docsvlist{#1}\ifx$#2$\else\,| #2\fi{:}{)}}% Process list
}
\newcommand{\switch}[1]{%
  \setcounter{totalcount}{0}% Reset total count
  \renewcommand*{\do}[1]{\stepcounter{totalcount}}% Reconfigure count
  \docsvlist{#1}% Count number of items
  \setcounter{listcount}{0}% Reset current item count
  \renewcommand*{\do}[1]{% Reconfigure item \do
    \stepcounter{listcount}% Next item
    ##1\ifnum\value{listcount}<\value{totalcount},\else\fi% Print item
  }
  \ensuremath{\left<\docsvlist{#1}\right>}% Process list
}

\begin{document}

\title{Interface Reconciliation in Kahn Process Networks using CSP and SAT}

\author{Pavel Zaichenkov, Olga Tveretina, Alex Shafarenko} \institute{Compiler
Technology and Computer Architecture Group,\\University of Hertfordshire,
United Kingdom\\ \email{\{p.zaichenkov,o.tveretina,a.shafarenko\}@ctca.eu}}

\maketitle

\begin{abstract}
We present a new CSP- and SAT-based approach for coordinating interfaces of
distributed stream-connected components provided as closed-source services.
The Kahn Process Network (KPN) is taken as a formal model of computation and
a Message Definition Language (MDL) is introduced to describe the format of
messages communicated between the processes.  MDL links input and output
interfaces of a node to support flow inheritance and contextualisation.  Since
interfaces can also be linked by the existence of a data channel between them,
the match is generally not only partial but also substantially nonlocal.  The
KPN communication graph thus becomes a graph of interlocked constraints to be
satisfied by specific instances of the variables.  We present an algorithm that
solves the CSP by iterative approximation while generating an adjunct Boolean
SAT problem on the way.  We developed a solver in \verb|OCaml| as well as tools
that analyse the source code of KPN vertices to derive MDL terms and
automatically modify the code by propagating type definitions back to the
vertices after the CSP has been solved.  Techniques and approaches are
illustrated on a KPN implementing an image processing algorithm as a running
example.
\end{abstract}

\keywords\
coordination programming, component programming, Kahn
Process Networks, interface coordination, constraint satisfaction,
satisfiability

\section{Introduction}
\label{sec:introduction}

The software intensive systems have reached unprecedented scale by every
measure: number of lines of code; number of people involved in the development;
number of dependencies between software components, and amount of data stored
and manipulated~\cite{uls}.  Many of them include heterogeneous elements, which
come from variety of different sources: parts of them are written in different
languages and tuned for different hardware/software platforms.  Furthermore,
when the software is developed and modified by dispersed teams, inconsistencies
in the design, implementation and usage are unavoidable.  This leads to clashes
of assumptions about operation cost, resource availability and algorithm
processing rate.  Last but not least, parts of the system are constantly
changing.  Many elements need to be replaced without negative effects on
performance or behaviour of the rest of the system.

One way to attack the software challenge is to suggest a component-based
design: a program is designed as a set of components, each represented by an
interface that specifies how they can be used in an application, and one or
more implementations which define their actual behaviour.  When a designer of
the application uses a component, they agree to rely only on the interface
specification.  Similarly, a developer who creates an implementation for
a component is unaware of the context where the component will be used.  An
algorithm that specifies the behaviour depends solely on self-contained input
and its result is produced in the form of a message without a specific
destination address.

The process network, introduced by G.~Kahn (KPNs)~\cite{Kahn1974}, is
a collection of stream-connected algorithmic building blocks, which are fully
independent single-threaded processes that lack a global state. The execution
of the network generally requires a supervisory coordination program that
manages the progress of the blocks and which provides a message-communication
infrastructure for the streams.  Since all domain-specific computations are
performed by the sequential processes inside the blocks, programming is
naturally separated into algorithm and concurrency
engineering~\cite{Jesshope-Shafarenko-Concurrency-Engineering}.  The
coordination language is responsible for component orchestration, namely
\begin{inparaenum}[1)]
    \item dynamic load control and adaptivity for a changing environment;
    \item access control to shared resources; and
    \item communication safety between components.
\end{inparaenum}
This paper focuses on the last aspect.  Component-based design requires an
implementation of a single component to be independent from the rest of the
network.  It raises a number of software engineering issues: components'
interfaces are required to be specific enough so that components are aware of
data structures communicated between them and, at the same time, generic enough
to facilitate decontextualisation and software reuse.

In this paper we present a solution to the interface reconciliation problem for
an interface definition language specifically designed for KPNs.  We
demonstrate a static mechanism (based on solving Constraint Satisfaction
Problem (CSP) and SAT) that checks compatibility of component interfaces
connected in a network with support of overloading and structural subtyping.
This allows one to design completely decontextualised components, so that they
may be reused in different contexts without changing the code.  This is
especially important when the components are provided as a compiled library and
its source is either private or unavailable.  The components are compatible
with a potentially unlimited number of input/output data formats coming from
the environment.  We also introduce a {\em flow
inheritance\/}~\cite{GrelSchoShafPPL08} mechanism: put simply, a message sent
from one component to another may also be required to contain additional data
which, although not needed by the recipient itself, is nevertheless required by
a component that the recipient sends its own messages to
(\Fig{fig:flow-inheritance}).

We propose a Message Definition Language (MDL) that enables components' generic
interfaces as well as subtyping and flow inheritance; we then recast the
interface reconciliation problem as a CSP for the interface variables and
propose an original solution algorithm that solves it by iterative
approximation while generating an adjunct Boolean SAT problem on the way.

We designed and implemented a communication protocol\footnote{for the avoidance
of doubt we state that the term ``protocol'' is used here in the sense of
`convention governing the structure and interpretation of messages' and not in
any state-transition sense} for components coded in \verb|C++| to demonstrate
the capabilities of MDL\@.  We developed tools that
\begin{inparaenum}[1)]
    \item automatically derive MDL interfaces from the source code;
    \item generate a set of constraints given a netlist\footnote{a textual
        representation of a graph} that describes the topology of the network;
    \item solve the CSP\@; and
    \item based on the solution of the CSP generate compilable code for
        every component with some API provided for run-time support.
\end{inparaenum}

The process is similar to template specialisation in \verb|C++|, however, in
our case, constraints that are produced by a pair of vertices may affect the
whole network, and, consequently, a global constraint satisfaction procedure is
required.  In this paper we provide a formal description of MDL, the CSP
definition and the algorithm designed to solve the CSP\@.

Throughout the paper we demonstrate the utility of the proposed approach on
a practical example: an image segmentation algorithm based on k-means
clustering (\Fig{fig:example}).

{\bf Related work.} Linda~\cite{Ahuja} is the first language designed to
separate computation and coordination models.  It is based on a simple
tuplespace model. One of the disadvantages of the model is that the knowledge
about the communication protocol is required while implementing the processes.
The problem of separation of concerns has not been solved in Concurrent
Collections from Intel~\cite{budimlic2010} (Linda's successor). Therefore,
generic components, which may be reused in multiple contexts without being
modified, are not supported in the language.

In the programming language Reo~\cite{Arbab} components are communicating
through hierarchical connectors that coordinate their activities and manipulate
message dataflow.  Similar to our approach, a constraint satisfaction engine,
which finds a solution that specifies a valid interaction between components,
is implemented.  S.~Kemper describes a SAT-based verification of Timed
Constraint Automata that is used for coordination of communicating
components~\cite{Kemper} as well as in Reo.  However, the research mostly
focuses on the design of reusable interaction protocols and lacks the
description of reusable component interfaces.

In previous years there were attempts to design efficient programming languages
and run-time systems for parallel programming based on
KPN~\cite{GrelSchoShafPPL08,nornir}, however, the interface reconciliation
problem stemming from nonlocal inheritance in KPNs has not been given enough
attention.

\section{Motivation}
\label{sec:motivation}

Kahn Process Networks is a concurrency model that introduces data streams in
the form of sequential channels that connect independent processes into
a network.  Decontextualisation of processes is an advantage of the  model.
Since processes do not share any data, a process's conformity with the context
is defined by its interface, which describes the kinds of message that the
process can send and receive.  Our goal is to provide a means of interface
coordination that supports genericity, i.e.\ the ability of an interface to
function correctly in a variety of contexts.

The distributed components are commonly provided as closed-source services.
Each service contains multiple processing functions compatible with a variety
of contexts. The input interface of a component is specialised based on the
message format that the message producer is capable to produce, and,
symmetrically, the output interface is specialised based on the consumer's
requirements.  For example, one can define a component that contains two
functions with type signatures \verb|Int -> Int| and \verb|Int -> String|.  The
functions implement algorithms that compute different values given the same
input.  The task is to statically choose the algorithm based the consumer's
requirements.  This also demonstrates that input and output types in the
interfaces of the KPN components are treated in the same manner.

The latter makes services fundamentally different from functions.  In
functional languages, a type signature that corresponds to the interface of the
component in the example can be defined by the intersection type \texttt{Int ->
Int $\land$ Int -> String}, which is unsound due to its ambiguity.  In
functional languages, the return type of a function depends solely on input
argument types. In contrast, the interfaces of the KPN components form
context-dependent relations that offer a selection of output types to
a consumer. This makes the typing decisions essentially nonlocal and genuinely
multiple.

The problem being solved can be seen as a type inference problem; however, it
cannot be solved using conventional type inference mechanisms based on
first-order unification due to the presence of polymorphic output types and
potential cyclic dependencies in the network (the example in~\Fig{fig:example}
contains a back edge).

%The existing methods do not solve the interface reconciliation problem for
%services and are applicable only to less flexible open-source components that
%are developed in a centralised manner. Support of cyclic networks makes the
%problem unsolvable by a simple forward algorithm.

\begin{figure}[t]
    \begin{center}
        \includegraphics[width=2.7in]{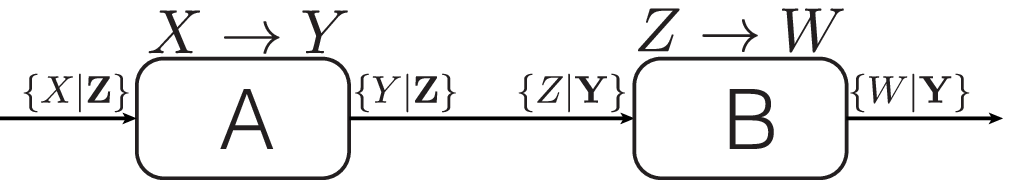}
    \end{center}
    %\vspace{-1em}
    \caption{Illustration of flow inheritance.  A component \texttt{A} can
        process a value of type $X$ and return a value of type $Y$ as a result.
        However, an input message contains not only an element of type $X$, but
        also an element of type $Z$.  The latter can be processed by
        a component \texttt{B}.  Flow inheritance provides a mechanism for
        partial message processing in a pipelined fashion.}
\label{fig:flow-inheritance}
\end{figure}

A common communication pattern in streaming networks is a pipeline, where
a message travels along a chain of components that work on its content.  The
component can accept a subtype of the input type, but part of the message may
be bypassed to another component down the pipeline if the message contains the
data the further component will need to use (\Fig{fig:flow-inheritance}).
Two modes of flow inheritance are envisaged, considered next.

{\bf Flow inheritance for records.} 
The fundamental type of a message in a variety of systems is {\em record},
which is a collection of label-value pairs.  Each component processes only
a specific set of pairs, however the pairs that the component does not require
may be bypassed to the output, so they can be processed in the next stages of
the pipeline if they are required.  For example, a message that represents
a geometric shape and has a type
\verb|{x:float, y:float, radius:float}|
may be processed in two steps: the first component processes the position of
the shape as defined by the pairs \verb|x| and \verb|y|, and the second one
needs the pair labelled \verb|radius|.

{\bf Flow inheritance for variants.} OOP extensively uses
overloading~\cite{strachey2000fundamental} to improve modularity and
reusability of code.  Similarly, support of polymorphic components in KPNs
facilitates their reuse in different contexts.  At the top level we see
a component's interface as a collection of alternative label- record pairs,
called {\em variants}, where the label corresponds to the particular
implementation that can process the message defined by the given record,
e.g.\\*[.5em]
\verb|(: cart: {x:float, y:float}, polar: {r:float, phi:float} :)|.\\*[.5em]
Here the colonised parentheses delimit the collection of variants and each
variant is associated with a record written as a set of label-value pairs.  Any
message that does not belong to one of the accepted variants must cause an
error unless there exists another component further down in the pipeline that
can process it.  In this case, the message should be bypassed to the recipient.

In streaming networks flow inheritance alleviates the problem and makes
configuration of individual interfaces independent from each other.  In our
work we developed a solution for the interface reconciliation based on the CSP
with support of flow inheritance for records and variants.  Our mechanism
statically detects implementation variants in components that are not required
in the context, which is important for applications running in the Cloud where
a user is charged proportionally to the amount of resources their application
uses.

\begin{figure}[t]
    \begin{center}
        \includegraphics[width=4in]{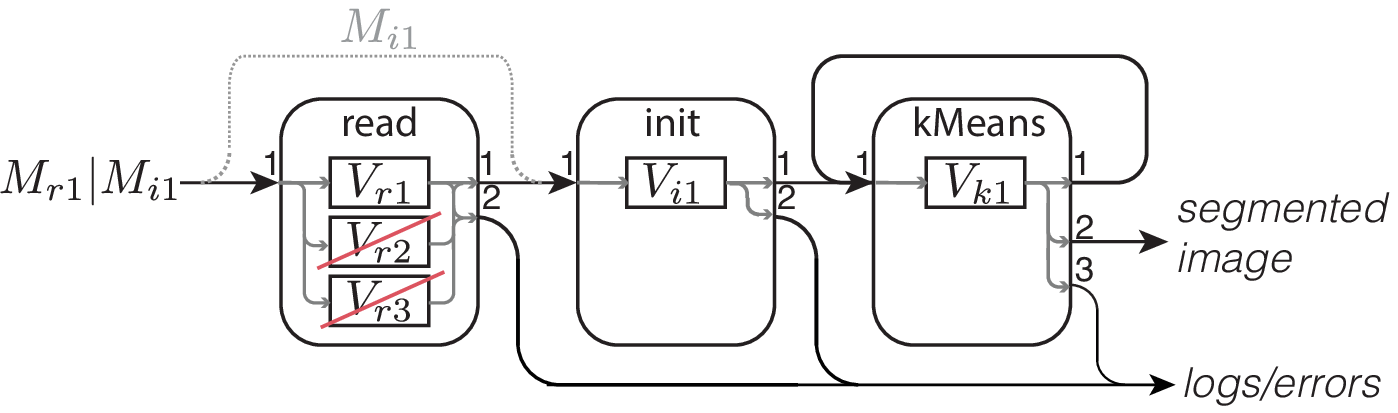}
    \end{center}
    \caption{An image segmentation algorithm based on k-means clustering that
    is implemented as a Kahn Process Network}
\label{fig:example}
\end{figure}

\label{sec:example}
{\bf Example.} As a running example we use our
implementation~\cite{repository-example} of an image segmentation algorithm
based on k-means clustering~\cite{macqueen1967}.  The applications's KPN graph
is shown in \Fig{fig:example}. The network represents a pipeline composed of
three components:
\begin{itemize}
    \item The component \verb|read| opens an image file using an input message
        $M_{r1}$ with the file name, and sends it to the first output channel.
        The component contains 3 functions that overload component's behaviour
        (i.e.\ the input interface of the component is defined by 3 variants):
        \begin{inparaenum}[1)]
            \item $V_{r1}$ loads the colour image in RGB format;
            \item $V_{r2}$ loads the greyscale image as an intensity one;
                and
            \item $V_{r3}$ loads the image as it is stored in the file.
        \end{inparaenum}
    \item The component \verb|init| sets initial parameters for the k-means
        algorithm.  The component contains one processing function $V_{i1}$.
        The input message can either come from the component \verb|read| or
        from the environment with an input message $M_{i1}$ if it has been
        opened and preprocessed before.  The input message must contain the
        number of clusters \verb|K| and the image itself.
    \item The \verb|kMeans| component represents an iterative implementation
        (defined as a function $V_{k1}$) of the k-means algorithm.  The result
        of each iteration is sent to the first output channel, which is
        circuited back to the input channel of the component itself.  This kind
        of design gives an opportunity to manage system load in the run-time
        and execute the next iteration only when sufficient resources are
        available.  Once the cluster centres have converged, the algorithm
        yields the result to the second output channel.
\end{itemize}

Using flow inheritance for variants $M_{i1}$ is routed directly to the
\verb|init| component bypassing a component \verb|read|.  Using flow
inheritance for records a parameter \verb|K| that is contained in $M_{r1}$ is
implicitly bypassed through \verb|read| to \verb|init|.

The interface reconciliation algorithm is capable of finding out that $V_{r2}$
and $V_{r3}$ are not used with the provided input, and functions containing the
implementations will be excluded from the generated code.

\section{Message Definition Language}
\label{sec:mdl}

Now  we define  the Message Definition Language (MDL) that describes component
interfaces.  Each component has its associated input and output interface
terms.  A {\em message\/} is a collection of data entities, defined by
a corresponding collection of {\em terms\/} that can contain term variables,
Boolean variables and Boolean expressions.

Each term is either atomic or a collection in its own right.  Atomic terms are
{\em symbols}, which are identifiers used to represent standard \verb|C++|
types, such as \verb|int| or \verb|string|.  To account for subtyping (including
the kinds that are not present in \verb|C++|) we include three categories of
collections (see \Fig{fig:semilattice}): {\em tuples\/} that demand exact match
and thus admit no structural subtyping, {\em records\/} that are subtyped
covariantly (a larger record is a subtype) and {\em choices\/} that are
contravariantly subtyped using set inclusion (a smaller choice is a subtype).
The intention of these terms is to represent
\begin{enumerate}
\item extensible data records~\cite{Gaster96apolymorphic,Leijen:scopedlabels},
    where additional named fields can be introduced without breaking the match
    between the producer and the consumer and where fields can also be
    inherited from input to output records by lowering the output type, which
    is always safe;
\item  data-record variants, where generally more variants can be accepted by
    the consumer than the producer is aware of, and where such additional
    variants can be inherited from the output back to the input of the producer
    --- hence contravariance --- again by raising the input type, which is
    always safe also.
\end{enumerate}

{\em Term variables\/} correspond to four categories of terms.  However, for
the correctness of the algorithm it is important to distinguish variables that
represent choices from variables that represent other term categories (due to
two kinds of subtyping defined by the seniority relation in
Definition~\ref{def:seniority}).  We use an {\em up-coerced\/} term variable,
e.g.~$\upvar{a}$, to represent a choice term and a  {\em down-coerced\/} term
variable, e.g.~$\downvar{a}$, to represent any other term, i.e.\ a symbol,
a tuple or a record.  Formally,

\vspace{-.5em}
\begin{grammar}
    <term variable> ::= $\uparrow$identifier | $\downarrow$identifier
\end{grammar}
\vspace{-.5em}

We  use symbol $\talloblong$ instead of $\uparrow$ or
$\downarrow$ symbols in the context where the sort is unimportant,
e.g.~$\var{a}$ is a term variable that can be either up-coerced or
down-coerced.

For brevity, term variables are called {\em variables}, Boolean variables are
called {\em flags\/} and Boolean expressions are called {\em guards}. The
following grammar specifies the guards:

\vspace{-.5em}
\begin{grammar}
    <bool> ::= (<bool> $\land$ <bool>) | (<bool> $\lor$ <bool>)
               | $\lnot$<bool> | "true" | "false" | flag
\end{grammar}
\vspace{-.5em}

MDL terms are built recursively using the constructors: tuple, record, choice
and switch, according to the following grammar:

\vspace{-.5em}
\begin{grammar}
    <term> ::=  <symbol> | <term variable> | <tuple> | <record> | <choice> | <switch>

    <tuple> ::= "("<term> $[$<term>$]^*$")"

    <record> ::= \hspace*{-1mm}"\{"$[$<label>"("<bool>"):"<term>$[$","<label>"("<bool>"):"<term>$]^* [$"|"$\downarrow$identifier $]]$"\}"

    <choice> ::= \hspace*{-1mm}"(:"$[$<label>"("<bool>"):"<term>$[$","<label>"("<bool>"):"<term>$]^*[$"|"$\uparrow$identifier $]]$":)"

    <label> ::= <symbol>

\end{grammar}
\vspace{-.5em}

Informally, a {\em tuple\/} is an ordered collection of terms and a {\em record\/} 
is an extensible, unordered collection of guarded labeled terms, where
{\em labels\/} are arbitrary symbols, which are unique within a single record.
A {\em choice\/} is a collection of alternative terms.  The syntax of choice is
the same as that of record except for the delimiters.  The difference between
records and choices is in subtyping and will become clear below when we define
seniority on terms.  We use choices to represent polymorphic messages
and component interfaces.

Records and choices are defined in {\em tail form}.  The tail is denoted by
a variable that represents a term of the same kind as the construct in which it
occurs.  For example, in the term
$\record{\el{l_1}{\true}{t_1},\dots,\el{l_n}{\true}{t_n}}{\downvar{v}}$ the
variable $\downvar{v}$ represents the tail of the record, i.e.~its members with
labels $l_i: l_i \neq l_1, \dots l_i \neq l_n$.  A {\em switch\/} is a set of
unlabeled (by contrast to a choice) guarded alternatives.

\vspace{-.5em}
\begin{grammar}
<switch> ::= "<"<bool>":"<term>$[$", "<bool>":"<term>$]^*$">"
\end{grammar}
\vspace{-.5em}

Exactly one guard must be \true{} for any valid switch.  The
switch is substitutionally equivalent to the term marked by the \true{} guard:

$$\switch{\el{}{\false}{t_1},
\dots,\el{}{\true}{t_i},\dots,\el{}{\false}{t_n}} = \switch{\el{}{\true}{t_i}}
= t_i\text{.}$$
The switch is an auxiliary construct intended for building conditional terms.  For example,
$\switch{\el{}{a}{int}, \el{}{\lnot a}{string}}$ represents the symbol $int$ if $a
= \true$, and the symbol $string$ otherwise.

For each term $t$, we use $\utvar(t)$ to denote the set of up-coerced term
variables that occur in $t$, $\dtvar(t)$ to denote the set of the down-coerced
ones, and $\bvar(t)$ to denote the set of flags.

A term $t$ is called {\em semi-ground\/} if it does not contain variables,
i.e.~$\utvar(t) \cup \dtvar(t) = \emptyset$.  A term $t$ is called {\em
ground\/} if it is semi-ground and does not contain flags, i.e.~$\utvar(t) \cup
\dtvar(t) \cup \bvar(t) = \emptyset$.

A term $t$ is {\em well-formed\/} if it is ground and one of the following
holds:
\vspace{-.7em}
\begin{enumerate}
    \item $t$ is a symbol.
    \item $n > 0$ and $t$ is a tuple \tuple{t_1,\dots,t_n} where all $t_i$, $0
        < i \leq n$, are well-formed.
    \item $n \geq 0$ and $t$ is a record
        \record{\el{l_1}{b_1}{t_1},\dots,\el{l_n}{b_n}{t_n}}{} where for all $0
        \leq i\ne j \leq n$, $b_i\wedge b_j \to l_i \neq l_j$ and all $t_i$ for
        which $b_i$ are \true{} are well-formed.
    \item $n \geq 0$ and  $t$ is a choice
        \choice{\el{l_1}{b_1}{t_1},\dots,\el{l_n}{b_n}{t_n}}{}  where for all
        $0 \leq i\ne j \leq n$, $b_i\wedge b_j \to l_i \neq l_j$ and all $t_i$
        for which $b_i$ are \true{} are well-formed.
    \item $n>0$ and $t$ is a switch $\switch{\el{}{b_1}{t_1},
        \dots,\el{}{b_n}{t_n}}$ where for some  $1\leq  i\leq n$, $b_i=\true$
        and $t_i$ is well-formed and where $b_j=\false$ for all $j \neq i$.
 \end{enumerate}

If an element of a record, choice or switch has a guard that is equal to
\false, then the element can be omitted, e.g.
$$\record{\el{l_1}{b_1}{t_1}, \el{l_2}{\false}{t_2},\el{l_3}{b_n}{t_3}}{}
= \record{\el{l_1}{b_1}{t_1}, \el{l_3}{b_3}{t_3}}{}\text{.}$$
If an element of a record or a choice has a guard that is \true, the guard can
be syntactically omitted, e.g.
$$\record{\el{l_1}{b_1}{t_1}, \el{l_2}{\true}{t_2},\el{l_3}{b_n}{t_3}}{}
= \record{\el{l_1}{b_1}{t_1}, \el{l_2}{}{t_2},\el{l_3}{b_n}{t_3}}{}\text{.}$$
We define the {\em canonical form\/} of a well-formed collection as
a representation that does not include \false{} guards, and we omit \true{}
guards anyway.  The canonical form of a switch is its (only) term with a
\true{} guard, hence any term in canonical form is switch-free.

Next we introduce a seniority relation on terms for the purpose of structural
subtyping.  In the sequel we use $\nil$ to denote the empty record
$\{\enskip\}$, which has the meaning of \texttt{void} type in \verb|C++| and
represents a message without any data.  Similarly, we use $\none$ to denote the
empty choice $(:\enskip:)$.

\begin{definition}[Seniority relation]
\label{def:seniority}
    The seniority relation $\rel$ on well-formed terms is defined in canonical
    form as follows:
    \begin{enumerate}
        \item $\none \rel t$ if $t$ is a choice.
        \item $t \rel \nil$ if $t$ is any term but a choice.
        \item $t \rel t$.
        \item $t_1 \rel t_2$, if for some $k,m>0$ one of the following holds:
        \begin{enumerate}
            \item $t_1 = \tuple{t^1_1,\dots,t^k_1}$, $t_2
                = \tuple{t^1_2,\dots,t^k_2}$ and $t^i_1 \rel t^i_2$ for each
                $1\leq i \leq k$;

            \item $t_1
                = \record{\el{l^1_1}{}{t^1_1},\dots,\el{l^k_1}{}{t^k_1}}{}$ and
                $t_2
                = \record{\el{l^1_2}{}{t^1_2},\dots,\el{l^m_2}{}{t^m_2}}{}$,
                where $k \geq m$ and for each $j \leq m$ there is $ i \leq k$
                such that $l^i_1 = l^j_2$ and $t^i_1 \rel t^j_2$;

            \item $t_1
                = \choice{\el{l^1_1}{}{t^1_1},\dots,\el{l^k_1}{}{t^k_1}}{}$ and
                $t_2
                = \choice{\el{l^1_2}{}{t^1_2},\dots,\el{l^m_2}{}{t^m_2}}{}$,
                where $k \leq m$ and for each $i \leq k$ there is $ j \leq m$
                such that $l^i_1 = l^j_2$ and $t^i_1 \rel t^j_2$;
        \end{enumerate}
    \end{enumerate}
\end{definition}

Given the relation $t \rel t'$, we say that $t'$ is {\em senior\/} to $t$ and
$t$ is {\em junior\/} to $t'$.

\begin{figure}[t]
    \begin{center}
        \begin{tikzpicture}[node distance=1.2cm,line width=.1mm]
    \node (nil) at (0,0) {\nil};
    \node [below of=nil] (tuple) {tuple};
    %\node [below of=tuple] (tupleEnd) {$\dots$};

    %\node [right of=tuple] (list) {list};
    %\node [below of=list] (listEnd) {$\dots$};
    \node [right of=tuple] (record) {record};
    \node [below of=record] (recordEnd) {$\dots$};

    %\node [left of=tuple] (integer) {integer};
    %\node [below of=integer] (integerEnd) {$\dots$};
    \node [left of=tuple] (symbol) {symbol};
    %\node [below of=symbol] (symbolEnd) {$\dots$};

    \node [right of=record] (choice) {choice};
    \node [above of=choice] (choiceEnd) {$\dots$};
    \node [below of=choice] (none) {\none};

    \node (dummy1) at (-2, 0) {};
    \node [below] (dummy2) at (-2, -2.2) {{\em subtype}};

    \draw (nil) -- (symbol);
    %\draw (nil) -- (integer);
    \draw (nil) -- (tuple);
    %\draw (nil) -- (list);
    \draw (nil) -- (record);
    \draw (choice) -- (none);

    %\draw (symbol) -- (symbolEnd);
    %\draw (integer) -- (integerEnd);
    %\draw (tuple) -- (tupleEnd);
    %\draw (list) -- (listEnd);
    \draw (record) -- (recordEnd);
    \draw (choice) -- (choiceEnd);

    \draw (choice) -- (choiceEnd);

    \draw[ultra thin,->] (dummy1) -- (dummy2);

\end{tikzpicture}
    \end{center}
    \caption{Two semilattices representing the seniority relation for terms of
    different categories.  The lower terms are the subtypes of the upper ones.}
\label{fig:semilattice}
\end{figure}
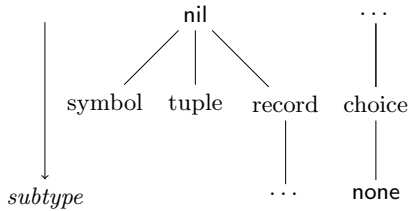

\begin{proposition}
\label{cor:semilattice}
    The seniority relation $\rel$ is trivially a partial order, and $(T, \rel)$
    is a pair of upper and lower semilattices (\Fig{fig:semilattice}).
\end{proposition}

The seniority relation represents the subtyping relation on terms.  If
a term $t'$ describes the input interface of a component, then the component
can process any message described by a term $t$, such that $t \rel t'$.

Although the seniority relation is straightforwardly defined for ground terms,
terms that are present in the interfaces of components can contain variables
and flags.  Finding such ground term values for the variables and such Boolean
values for the flags that the seniority relation holds represents a CSP
problem, which is formally introduced next.

\section{Constraint Satisfaction Problem for KPN}
\label{sec:problem-def}

In this section we define a Constraint Satisfaction Problem for Kahn Process
Networks (\cspkpn{}).  We regard a KPN network as a directed weakly connected
labeled graph $\g=(\ver, \ed,\lab)$, where

\begin{enumerate}
    \item $\ver$ is a set of vertices.  The vertices correspond to individual
        Kahn processes.

    \item $\ed$ is a set of edges, where each edge $e \in \ed$ is an ordered
        pair of vertices $(v,v')$, $v,v' \in \ver$.  The edges correspond to
        channels connecting Kahn processes.

    \item A function $\lab\colon \ed \rightarrow \term \times \term$ assigns
        a label\footnote{we use the same word ``label'' to refer to the mark on
        a graph edge and the symbol that labels a term in a record or a choice;
        however our intention is always clear from the context.} to each edge
        $e \in \ed$ which represents a pair of MDL terms
        $\lab(e)=\constr{t}{t'}$ called a {\it constraint\/}\footnote{in the rest
        of the paper symbol $\rel$ denotes the seniority relation for a pair
        ground terms; alternatively, if the terms are not ground, $\rel$
        specifies a constraint.}.  It defines the input requirements and the
        output properties associated with the channel.
\end{enumerate}

Given a graph $\g=(\ver,\ed,\lab)$ we define the set of constraints as
$$\con{\g}=\bigcup_{e\in \ed} \lab(e)\text{,}$$
the sets of up-coerced term variables $\utvars(\g)$ and down-coerced term
variables as $\dtvars(\g)$
$$
\utvars(\g) = \bigcup_{\constr{t}{t'} \in \con{\g}}\utvar(t) \cup \utvar(t')
    \quad\text{and}\quad
\dtvars(\g) = \bigcup_{\constr{t}{t'} \in \con{\g}}\dtvar(t) \cup \dtvar(t')\text{,}
$$
and the set of flags $\bvars(\g)$ as
$$\bvars(\g) = \bigcup_{\constr{t}{t'} \in \con{\g}}\bvar(t) \cup \bvar(t')\text{.}$$
Assume a vector of flags $\vec{f}=(f_1,\dots,f_l)$, a vector of term variables
$\talloblong{\vec{v}}=(\talloblong{v_1},\dots,\talloblong{v_m})$, a vector of
Boolean values $\vec{b}=(b_1,\dots,b_l)$ and a vector of terms
$\vec{s}=(s_1,\dots, s_m)$.
Then for each term $t$
\begin{enumerate}
    \setlength\itemsep{.1em}
    \item $t[\vec{f}/\vec{b}]$ denotes the vector obtained as a result of the
        simultaneous replacement of $f_i$ with $b_i$ for each $1\leq j \leq l$;
    \item $t[\talloblong{\vec{v}}/\vec{s}]$ denotes the vector obtained as
        a result of the simultaneous replacement of $\talloblong{v}_i$ with
        $s_i$ for each $1\leq i\leq m$;
    \item $t[\vec{f}/\vec{b},\talloblong{\vec{v}}/\vec{s}]$ is a shortcut for
        $t[\vec{f}/\vec{b}][\talloblong{\vec{v}}/\vec{s}]$.
\end{enumerate}

Assume a KPN graph $\g=(\ver, \ed, \lab)$ such that $|\bvars(\g)|=l$,
$|\utvars(\g)|=m$, $|\dtvars(\g)|=n$ and for some $l,m,n \geq 0$.
\begin{definition}[\cspkpn{}]
    We define a CSP for a KPN graph $\g$ (CSP-KPN) as follows: for each
    $\constr{t}{t'} \in \con{\g}$ find a vector of Boolean values $\vec{b}
    = (b_1,\dots,b_l)$, a vectors of ground terms $\vec{t}=(t_1,\dots,t_m)$,
    $\vec{t'}=(t'_1,\dots,t'_n)$ such that
    $$t[\vec{f}/\vec{b},\upvar{\vec{v}}/\vec{t},\downvar{\vec{v}}/\vec{t}'] \rel t'[\vec{f}/\vec{b},\upvar{\vec{v}}/\vec{t},\downvar{\vec{v}}/\vec{t}']\text{,}$$
    where $\vec{f}=(f_1,\dots,f_l)$,
    $\upvar{\vec{v}}=(\upvar{v_1},\dots,\upvar{v_m})$,
    $\downvar{\vec{v}}=(\downvar{v_1},\dots,\downvar{v_n})$.  The tuple
    $(\vec{b}, \vec{t}, \vec{t}')$ is called a solution.
\end{definition}

A \cspkpn{} is decidable since the message definition language we introduced
can be seen as a term algebra, and decision problems for term algebras are
decidable~\cite{book-cclt}.

\section{Adjunct SAT}
\label{sec:adjunct-sat}

The \cspkpn{} solution algorithm presented in the next section is iterative and
takes advantage of the order-theoretical structure of the MDL
(Proposition~\ref{cor:semilattice}).

Let $\ssat_0\subseteq \ssat_1\subseteq \dots \subseteq \ssat_s$ be sets of
Boolean constraints, and  $\vec{a}^\uparrow$ and $\vec{a}^\downarrow$ be
vectors of semiground terms such that $|\vec{a}^\uparrow| = |\utvars(\g)|$ and
$|\vec{a}^\downarrow| = |\dtvars(\g)|$.  The vectors $\vec{a}^\uparrow$ and
$\vec{a}^\downarrow$ are {\em conditional approximations\/} of the solution.

We seek the solution as a fixed point of a series of approximations in the
following form:
$$
    (\ssat_0, \vec{a}^\uparrow_0, \vec{a}^\downarrow_0), \ldots,
    (\ssat_{s-1}, \vec{a}^\uparrow_{s-1}, \vec{a}^\downarrow_{s-1}),
    (\ssat_s, \vec{a}^\uparrow_s, \vec{a}^\downarrow_s)\text{,}
$$
where for every $1\leq k \le s$ and a vector of Boolean values $\vec{b}$ that
is a solution to $\mathsf{SAT}(\ssat_k)$ (by $\mathsf{SAT}(\ssat_k)$ we mean
a set of Boolean vector satisfying $\ssat_k$):

\begin{equation}
    \label{eq:propterms}
    \vec{a}^\uparrow_{k-1}[\vec{f}/\vec{b}] \rel \vec{a}^\uparrow_{k}[\vec{f}/\vec{b}]
    \qquad\text{and}\qquad
    \vec{a}^\downarrow_{k}[\vec{f}/\vec{b}] \rel \vec{a}^\downarrow_{k-1}[\vec{f}/\vec{b}]\text{,}
\end{equation}
where the elements of the vectors are compared pairwise.  The starting point is
$\ssat_0=\emptyset$, $\vec{a}^\uparrow_0=(\none,\dots,\none)$,
$\vec{a}^\downarrow_0=(\nil,\dots,\nil)$ and the series terminates as soon as
$\mathsf{SAT}(\ssat_s) = \mathsf{SAT}(\ssat_{s-1})$,
$\vec{a}^\uparrow_{s} = \vec{a}^\uparrow_{s-1}$,
$\vec{a}^\downarrow_{s} = \vec{a}^\downarrow_{s-1}$.

The adjunct set of Boolean constraints potentially expands at every iteration
of the algorithm by inclusion of further logic formulas called {\em
assertions\/} into its conjunction as the algorithm processes constraints
$\con{G}$.  Whether the set of Boolean constraints actually expands or not can
be determined by checking the satisfiability of $\mathsf{SAT}(\ssat_k)\neq
\mathsf{SAT}(\ssat_{k-1})$ for the current iteration $k$.

We argue below that if the original \cspkpn{} is satisfiable then so is
$\mathsf{SAT}(\ssat_s)$ and that the tuple of vectors
$(\vec{b_s}, \vec{a}^\uparrow_s[\vec{f}/\vec{b_s}], \vec{a}^\downarrow_s[\vec{f}/\vec{b_s}])$
is a solution to the former, where $\vec{b_s}$ is a solution of
$\mathsf{SAT}(\ssat_s)$.  In other words, the iterations terminate when the
conditional approximation limits the term variables, and when the adjunct SAT
constrains the flags enough to ensure the satisfaction of all \cspkpn{}
constraints.
In general, the set $\mathsf{SAT}(\ssat_s)$ can have more than one solution.
We select one of them, denoted by $\mathsf{SATSol}(\ssat_s)$ in the algorithm.

\begin{algorithm}[t]
    \caption{$\cspkpn(\g)$}
    \label{alg:cspkpn}
\begin{algorithmic}[1]
    \State$c \gets |\con{\g}|$
    \State$i \gets 0$
    \State$\ssat_{0} \gets \emptyset$
    \State$\vec{a}^\uparrow_{0} \gets (\none,\dots,\none)$
    \State$\vec{a}^\downarrow_{0} \gets (\nil,\dots,\nil)$
    \Repeat \label{alg:cspkpn-loop-start}
    \For{$1 \leq j \leq c:\ \constr{t_j}{t'_j} \in \con{\g}$}
        \State$(\ssat_{i \cdot c + j}, \vec{a}^\uparrow_{i \cdot c + j}, \vec{a}^\downarrow_{i \cdot c + j}) \gets
            \solve(\ssat_{i \cdot c + j-1}, \vec{a}^\uparrow_{i \cdot c + j-1}, \vec{a}^\downarrow_{i \cdot c + j-1}, \true, t_j, t'_j)$
    \EndFor
    \State$i \gets i + 1$
    \Until{$(\mathsf{SAT}(\ssat_{i \cdot c}), \vec{a}^\uparrow_{i \cdot c}, \vec{a}^\downarrow_{i \cdot c}) = (\mathsf{SAT}(\ssat_{(i-1) \cdot c}), \vec{a}^\uparrow_{(i-1) \cdot c}, \vec{a}^\downarrow_{(i-1) \cdot c})$}
    \If{$\ssat_{i \cdot c}$ is unsatisfiable}
        \State \Return$\unsat$
    \Else
        \State \Return$(\mathsf{SATSol}(\ssat_{i \cdot c}), \vec{a}^\uparrow_{i \cdot c}[\vec{f}/\vec{b}], \vec{a}^\downarrow_{i \cdot c}[\vec{f}/\vec{b}])$
    \EndIf
\end{algorithmic}
\end{algorithm}

\section{Algorithm}
\label{sec:algorithm}

In this section we present Algorithm~\ref{alg:cspkpn} which solves CSP-KPN for
a given KPN graph $\g=(\ver,\ed,\lab)$.  It performs the following steps.

The algorithm iterates over the set of constraints $\con{\g}$ and at each step
it builds a closer approximation of the solution.  The relation between two
consequent approximations satisfies formulas (\ref{eq:propterms}).

The function $\solve()$ solves the constraint $\constr{t_j}{t'_j}$ (see
equation~\eqref{eq:alg-equation} in Lemma~\ref{lemma-series}) and updates the
vectors
$\vec{a}^\uparrow_{i\cdot c + j}$ and $\vec{a}^\downarrow_{i \cdot c + j}$
with new values.  Furthermore, it adds Boolean assertions presented below that
ensure 1) satisfaction of the constraint for any
$\vec{b} \in \mathsf{SAT}(\ssat_{i \cdot c + j})$ as provided by
Definition~\ref{def:seniority}; and 2) well-formedness of the terms occurring
in it.

The algorithm terminates if $\ssat_{i \cdot c} \equiv \ssat_{(i - 1) \cdot c}$,
$\vec{a}^\uparrow_{i \cdot c}= \vec{a}^\uparrow_{(i - 1) \cdot c}$ and
$\vec{a}^\downarrow_{i \cdot c}= \vec{a}^\downarrow_{(i - 1) \cdot c}$.\\*[1em]
{\bf Well-formedness assertions for records and choices.} Any pair of elements
in a well-formed record/choice cannot have equal labels.  Therefore, for each
record $\record{\el{l_1}{b_1}{t_1},\dots,\el{l_1}{b_n}{t_n}}{}$ and each choice
$\choice{\el{l_1}{b_1}{t_1},\dots,\el{l_1}{b_n}{t_n}}{}$ occurring anywhere in
$\con{\g}$ the following assertion is added to the SAT\@:
$$\bigwedge_{\forall 1 \leq i, j \leq n\colon l_i = l_j} \lnot (b_i \wedge b_j)\text{.}$$
{\bf Well-formedness assertions for switches.} A well-formed switch term must
have exactly one positive guard.  Hence, for each switch
$\switch{\el{}{b_1}{t_1},\dots,\el{}{b_n}{t_n}}$ occurring anywhere in
$\con{\g}$ the following assertion is added to the SAT\@:
$$(b_1 \vee \dots \vee b_n) \wedge \bigwedge_{\forall 1 \leq i, j \leq n\colon i \neq j} \lnot (b_i \wedge b_j)\text{.}$$
{\bf Order assertions.} We generate two kinds of order assertions.
\begin{enumerate}
    \item If a variable $\var{x}$ is junior to two
        incommensurable, identically guarded terms
        $\var{x} \rel \switch{\el{}\ldots b:t_1\ldots}$ and
        $\var{x} \rel \switch{\el{}\ldots b:t_2\ldots}$, where neither
        $t_1 \rel t_2$ nor $t_2 \rel t_1$, the assertion $\lnot b$ is added to
        the adjunct SAT\@.

    \item For each $c\in\con{\g}$ of the form
        $\constr{
            \switch{\el{}{b_1}{t_1},\dots,\el{}{b_n}{t_n}}}
            {\switch{\el{}{b'_1}{t'_1},\dots,\el{}{b'_m}{t'_m}}}\text{,}
        $
        the assertion $\lnot (b_i \land b'_j)$ is added to the adjunct SAT
        if $t_i \not \rel t'_j$.
\end{enumerate}

Further details are found in Appendix~\ref{sec:appendix}.%~\cite{}.
\begin{lemma}[Loop invariant]
\label{lemma-series}
Algorithm~\ref{alg:cspkpn} finds a series of approximations in the form of
$$
    (\ssat_{k_0}, \vec{a}^\uparrow_{k_0}, \vec{a}^\downarrow_{k_0}), \ldots,
    (\ssat_{k_{s-1}}, \vec{a}^\uparrow_{k_{s-1}}, \vec{a}^\downarrow_{k_{s-1}}),
    (\ssat_{k_s}, \vec{a}^\uparrow_{k_s}, \vec{a}^\downarrow_{k_s})\text{,}
$$
where $k_{i} = i \cdot |\con{\g}|$, and such that the following holds for any
$\vec{b}_{k_{i}}\in \mathsf{SAT}(\ssat_{k_{i}})$.
\begin{enumerate}
    \item $\ssat_{k_{i}} \supseteq \ssat_{k_{i-1}}$;
    \item $\vec{a}^\uparrow_{k_{i-1}}[\vec{f}/\vec{b}_{k_{i}}] \rel
        \vec{a}^\uparrow_{k_{i}}[\vec{f}/\vec{b}_{k_{i}}]$ and
        $\vec{a}^\downarrow_{k_{i}}[\vec{f}/\vec{b}_{k_{i}}] \rel
        \vec{a}^\downarrow_{k_{i-1}}[\vec{f}/\vec{b}_{k_{i}}]$;
    \item $\nexists\ (\vec{a}', \vec{a}'')$:
        \begin{enumerate}
            \item
                $\vec{a}^\uparrow_{k_{i-1}}[\vec{f}/\vec{b}_{k_{i}}] \rel
                \vec{a}'[\vec{f}/\vec{b}_{k_{i}}]$,
                $\vec{a}'[\vec{f}/\vec{b}_{k_{i}}] \rel
                \vec{a}^\uparrow_{k_{i}}[\vec{f}/\vec{b}_{k_{i}}]$;
            \item
                $\vec{a}^\downarrow_{k_{i}}[\vec{f}/\vec{b}_{k_{i}}] \rel
                \vec{a}''[\vec{f}/\vec{b}_{k_{i}}]$,
                $\vec{a}''[\vec{f}/\vec{b}_{k_{i}}] \rel
                \vec{a}^\downarrow_{k_{i-1}}[\vec{f}/\vec{b}_{k_{i}}]$.
        \end{enumerate}
\end{enumerate}
\end{lemma}

\begin{proof}
Let $c = |\con{\g}|$. To construct $(\ssat_{k_i}, \vec{a}^\uparrow_{k_i},
\vec{a}^\downarrow_{k_i})$ given $(\ssat_{k_{i-1}}, \vec{a}^\uparrow_{k_{i-1}},
\vec{a}^\downarrow_{k_{i-1}})$, the algorithm iteratively calls  $\solve()$
function (see Appendix~\ref{sec:appendix}).% in~\cite{})
$$
    (\ssat_{r}, \vec{a}^\uparrow_{r}, \vec{a}^\downarrow_{r}) \gets \solve(\ssat_{r-1}, \vec{a}^\uparrow_{r-1}, \vec{a}^\downarrow_{r-1}, \true, t_j, t'_j)
$$
where $r = k_{i-1}+j$, $1 \leq j \leq c$ and $k_i = k_{i-1} + c$.  For any $t_j
\rel t'_j$, $\solve()$ constructs the Boolean constraints $\ssat_{r}$, and
finds $\vec{a}^\uparrow_{r}$ and $\vec{a}^\downarrow_{r}$ by solving the
equation
\begin{equation}
    \label{eq:alg-equation}
    t[\vec{f}/\vec{b}_{r-1},\upvar{\vec{v}}/\vec{a}^\uparrow_{r-1},\downvar{\vec{v}}/\vec{a}^\downarrow_{r}] = t'[\vec{f}/\vec{b}_{r-1},\upvar{\vec{v}}/\vec{a}^\uparrow_{r},\downvar{\vec{v}}/\vec{a}^\downarrow_{r-1}]\text{,}
\end{equation}

\begin{enumerate}
    \item $\ssat_{k_{i}} \supseteq \ssat_{k_{i-1}}$ by construction:
        $\solve()$ only adds new Boolean constraints to the existing
        set.
    \item $\solve()$ iteratively constructs the local approximation
        for each constraint. The series of local approximations converges to
        the global approximation.
    \item Proof by contradiction. Assume that $(\vec{a}^\uparrow_{r},
        \vec{a}^\downarrow_{r})$ is a solution of \eqref{eq:alg-equation} and
        there exists another solution $(\vec{a}', \vec{a}'')$, such that
        $\vec{a}' \neq \vec{a}^\uparrow_{r}$ and $\vec{a}'' \neq
        \vec{a}^\downarrow_{r}$. Then
        $$
        t[\vec{f}/\vec{b}_{r-1},\upvar{\vec{v}}/\vec{a}^\uparrow_{r-1},\downvar{\vec{v}}/\vec{a}''] = t'[\vec{f}/\vec{b}_{r-1},\upvar{\vec{v}}/\vec{a}',\downvar{\vec{v}}/\vec{a}^\downarrow_{r-1}]\text{.}
        $$
        Two ground terms are equal only if they represent the same term, and,
        therefore, $\vec{a}' = \vec{a}^\uparrow_{r}$ and $\vec{a}''
        = \vec{a}^\downarrow_{r}$, which contradicts the initial assumption.\qed
\end{enumerate}
\end{proof}

\begin{theorem}[Termination]\label{thm:termination}
    $\cspkpn(\g)$ terminates after a finite number of steps for any KPN graph
    $\g$.
\end{theorem}

\begin{proof}
For a given graph $\g$ the number of flags, variables and labels for records
and choices is bounded. There are two ways to produce new terms: either to add
entries with new labels to records and choices, or to substitute subterms for
terms.
\begin{enumerate}
    \item The number of new terms constructed by adding new entries is bounded
        because the number of labels in a given $\g$ is finite.
    \item The number of terms constructed by substituting subterms for other
        terms is bounded because
        \begin{inparaenum}[a)]
            \item the number of variables is finite (the algorithm does not
                generate new variables);
            \item after the variables have been instantiated, the category of
                the term cannot be changed, otherwise, the seniority relation
                would be violated.
        \end{inparaenum}
\end{enumerate}
It implies that for each $\upvar{v} \in \utvars(\g)$ there exists a ground term
$\bar{t}$, such that $\upvar{v} \rel \bar{t}$, and for each $\downvar{v} \in
\dtvars(\g)$ there exists a ground term $\underline{t}$, such that
$ \underline{t} \rel \downvar{v}$.
Providing that $|\mathsf{SAT}(\ssat_{k_i})| \leq
|\mathsf{SAT}(\ssat_{k_{i-1}})|$, the algorithm terminates after a finite
number of steps.\qed
\end{proof}

\begin{theorem}
    Assume a KPN graph $\g=(\ver, \ed, \lab)$. The set of constraints $\con{\g}$
    is inconsistent if and only if $\cspkpn(\g)$ returns $\unsat$.
\end{theorem}

\begin{proof}
As the initial approximation the algorithm selects the weakest approximation
$(\emptyset, (\none,\dots,\none), (\nil,\dots,\nil))$,
it follows from Lemma~\ref{lemma-series} that the algorithm iterates over all
possible approximations in consecutive order starting from
$(\ssat_{k_0}, \vec{a}^\uparrow_{k_0}, \vec{a}^\downarrow_{k_0})$. Therefore,
the algorithm cannot skip a solution if one exists.  By
Theorem~\ref{thm:termination}  the algorithm terminates after a finite number
of steps.  Hence, it returns $\unsat$ only if and only if the set of
constraints $\con{\g}$ cannot be satisfied.\qed
\end{proof}

\section{Communication Protocol}
\label{sec:protocol}

In this section we demonstrate interfaces with flow inheritance and code
customisation using the example from Section~\ref{sec:example}.  The interfaces
are defined as choice-of-records terms.  Labels in the choice term of the input
interface correspond to function names that can process messages tagged with
corresponding labels.  Output messages are produced by calling special
functions called {\em salvos}.  The name of a salvo corresponds to one of the
labels in the output choice term.  The compatibility of two components
connected by a channel is defined by the seniority relation.

Consider the source code and the interface of the component \verb|read| in
\Fig{fig:example-demo}.  Integers that have been added as prefixes to functions
in the code specify the channels that messages are received from and sent to.
In the interfaces we use prefixes \verb|$^| and \verb|$_| before identifiers to
denote up- and down-coerced variables, respectively.

\begin{figure}[t]
\centering
\begin{subfigure}[b]{\textwidth}
\begin{verbatim}
message _1_init(vector<vector<double> img);
message _2_error(string msg);
variant _1_read_color(string fname) { _1_init(...); ... _2_error(...); }
variant _1_read_grayscale(string fname) {...}
variant _1_read_unchanged(string fname) {...}
\end{verbatim}
\caption{The source code}
\end{subfigure}
\begin{subfigure}[b]{\textwidth}
\begin{verbatim}
IN
  1: (: read_color(c): {fname: string | $_rc},
        read_grayscale(g): {fname: string | $_rg},
        read_unchanged(u): {fname: string | $_ru} | $^r :)
OUT
    1: (: init(or c g u): {img: vector<vector<double>>, | $_ro1 } | $^r :)
    2: (: error(or c g u): {msg: string | $_ro2 } :)
$_rc <= $_ro1; $_rg <= $_ro1; $_ru <= $_ro1;
$_rc <= $_ro2; $_rg <= $_ro2; $_ru <= $_ro2;
\end{verbatim}
\caption{The interface}
\end{subfigure}
\caption{The source code and the interface of the component \texttt{read} of
the image processing algorithm}
\label{fig:example-demo}
\end{figure}

A tail variable \verb|$^r| in the interface enables flow inheritance for
choices: variants from the input channel that cannot be processed by the
component (i.e.\ all variants but \verb|read_color|, \verb|read_grayscale| or
\verb|read_unchanged|), are absorbed by \verb|$^r|.  Thus, the messages of type
$M_{r1}$ that contain the name of the image file are processed by the component
and the messages of type $M_{i1}$ are inherited to the output and forwarded
directly to the component \verb|init|.

Flow inheritance for records is realised by down-coerced variables \verb|$_rc|,
\verb|$_rg|, \verb|$_ru|, \verb|$_ro1|, \verb|$_ro2|, and a set of auxiliary
constraints.  A record in the input message contains an entry with the label
\verb|K|, which a processing function does not expect.  After solving the CSP,
the entry is added to the tail variable \verb|$_ro1|, because the solver
deduces that the element with the label \verb|K| is required by the component
\verb|init|.

Furthermore, we use flags \verb|c|, \verb|g| and \verb|u| to exclude the code
that is not used in the context.  The guards in the output interface employ the
joint set of flags from the input variants that can fire salvos specified in
the output interface.  In the example all three functions can fire \verb|init|
and \verb|error| salvos; accordingly, the salvos' guards are $c \lor g \lor u$.
The solver deduces that that the variants \verb|read_grayscale| and
\verb|read_unchanged| cannot receive any messages, and, therefore, their
respective processing functions can be excluded from the code.

To facilitate decontextualisation we introduce a wrapper for every component
called a {\em shell\/}: an auxiliary configuration file that provides
facilities for renaming labels in output records and choices and changing the
routing of output messages.

The source code and the interfaces for the other two components are available in
the repository~\cite{repository-example}.

\section{Implementation}
\label{sec:implementation}

We implemented the solver~\cite{repository} for the $\cspkpn{}$ in
\verb|OCaml|.  It works on top of the PicoSAT~\cite{picosat} library, the
latter used as a subsolver dealing with Boolean assertions.  The input for the
solver is a set of constraints and the output is in the form of assignments to
flags and term variables.

We also developed a toolchain in \verb|C++| and \verb|OCaml| that performs the
interface reconciliation in five steps:
\begin{enumerate}
    \item Given a set of \verb|C++| sources (the components),
        augment them with macros acting as placeholders
        for the code that enables flow-inheritance.
    \item Derive the interfaces from the code.
    \item Given the interfaces and a netlist that specifies a KPN graph,
        construct the constraints to be passed on to the \cspkpn{}
        algorithm.
    \item Run the solver.
    \item Based on the solution, generate header files for every component with
        macro definitions.  In addition, the tool generates
        the API functions to be called when a component sends
        or receives a message.
\end{enumerate}

Advantages of the presented design are the following:
\begin{itemize}
    \item Interfaces and the code behind them can be generic as long as they
        are sufficiently configurable. No communication between component
        designers is necessary to ensure consistency in the design.
    \item Configuration and compilation of every component is separated from
        the rest of the application.  This prevents source code leaks in
        proprietary software running in the Cloud.\footnote{which is otherwise
        a serious problem. For example, proprietary \texttt{C++} libraries that
        use templates cannot be distributed in binary form due to restrictions
        of the language's static specialisation mechanism.}
\end{itemize}

\section{Conclusion and Future Work}
\label{sec:conclusion}

We have presented a new static mechanism for coordinating component interfaces
based on CSP and SAT that checks compatibility of component interfaces
connected in a network with support of overloading and structural subtyping.
We developed a fully decoupled Message Definition Language that can be used in
the context of KPN for coordinating components written in any programming
language.  We defined the interface of  \verb|C++| components to demonstrate
the binding between the MDL and message processing functions.  Our techniques
support genericity, inheritance and structural subtyping, thanks to the order
relation defined on MDL terms.

On the theory side, we presented the CSP solution algorithm, showed its
correctness and identified the termination condition.  Although we assume that
the algorithm is NP-complete because of the SAT problem, which needs to be
solved as a subproblem, the complexity of the algorithm will be evaluated in
further research.

The next step will be to support multiple flow inheritance in the MDL, to
enable combined structures with inheritance (for example, $\tuple{union\
\downvar{a}\ \downvar{b}}$ represents a record that contains a union of entries
associated with records $\downvar{a}$ and $\downvar{b}$).  This would allow one
to design components that perform synchronisation and merge multiple messages
into one while preserving the inheritance mechanism of a vertex.

In the context of Cloud, our results may prove useful to the
software-as-service community since we can support much more generic interfaces
than are currently available without exposing the source code of proprietary
software behind them. Building KPNs the way we do could enable service
providers to configure a solution for a network customer based on components
that they have at their disposal as well as those provided by other providers
and the customer themselves, all solely on the basis of interface definitions
and automatic tuning to nonlocal requirements.

\bibliographystyle{splncs}
\bibliography{bibliography}

\newpage
\appendix
\section{Appendix: $\solve$ function}
\label{sec:appendix}

\begin{algorithm}[H]
    \caption{$\solve(\ssat_i, \vec{a}^\uparrow_i, \vec{a}^\downarrow_i, b, t, t')$}
\label{alg:solve}
\begin{algorithmic}[1]
    \State$\hat{\ssat}_i \gets \assertWellFormed(\assertWellFormed(\ssat_i, b, t), b, t')$
    \If{(($t = \none$ or $t' = \nil$) and ($t \neq \none$ or $t' \neq \nil$))\\\qquad or ($t$ and $t'$ are ground, and $t = t'$)}
        \State\Return$(\hat{\ssat}_i, \vec{a}^\uparrow_i, \vec{a}^\downarrow_i)$
    \ElsIf{$t = \upvar{v}$ and $t' = \upvar{v'}$}
        \State$\ssat_{i+1}, \vec{a}^\uparrow_{i+1} \gets \text{set a new approximation for $\upvar{v'}$ equal to the one of $\upvar{v}$}$
    \ElsIf{$t = \downvar{v}$ and $t' = \downvar{v'}$}
        \State$\ssat_{i+1}, \vec{a}^\downarrow_{i+1} \gets\text{set a new approximation for $\downvar{v}$ equal to the one of $\downvar{v'}$}$
        \State\Return$(\ssat_{i+1}, \vec{a}^\uparrow_i, \vec{a}^\downarrow_{i+1})$
    \ElsIf{$t$ is $\nil$, symbol, tuple or record, and $t' = \downvar{v'}$}
        \State\Return$\solve(\hat{\ssat}_i, \vec{a}^\uparrow_i, \vec{a}^\downarrow_i, b, t, t'[\upvar{\vec{v}}/\vec{a}^\uparrow_i,\downvar{\vec{v}}/\vec{a}^\downarrow_i])$
    \ElsIf{$t$ is a choice and $t' = \upvar{v'}$}
        \State$\ssat_{i+1}, \vec{a}^\uparrow_{i+1} \gets \text{set a new approximation for $\upvar{v'}$ as $t[\upvar{\vec{v}}/\vec{a}^\uparrow_i,\downvar{\vec{v}}/\vec{a}^\downarrow_i]$ when $b$}$
        \State\Return$(\ssat_{i+1}, \vec{a}^\uparrow_{i+1}, \vec{a}^\downarrow_i)$
    \ElsIf{$t = \downvar{v}$, and $t'$ is $\nil$, symbol, tuple or record}
        \State$\ssat_{i+1}, \vec{a}^\downarrow_{i+1} \gets \text{set a new approximation for $\downvar{v}$ as $t'[\upvar{\vec{v}}/\vec{a}^\uparrow_i,\downvar{\vec{v}}/\vec{a}^\downarrow_i]$ when $b$}$
        \State\Return$(\ssat_{i+1}, \vec{a}^\downarrow_{i}, \vec{a}^\downarrow_{i+1})$
    \ElsIf{$t = \upvar{v}$ and $t'$ is a choice}
        \State\Return$\solve(\hat{\ssat}_i, \vec{a}^\uparrow_i, \vec{a}^\downarrow_i, b, t[\upvar{\vec{v}}/\vec{a}^\uparrow_i,\downvar{\vec{v}}/\vec{a}^\downarrow_i], t')$
    \ElsIf{$t$ and $t'$ are tuples}
        \State\Return$\solveTupleTuple(\hat{\ssat}_i, \vec{a}^\uparrow_i, \vec{a}^\downarrow_i, b, t, t')$
    \ElsIf{$t = \nil$ and $t'$ is a record}
        \State\Return$\solveNilRecord(\hat{\ssat}_i, \vec{a}^\uparrow_i, \vec{a}^\downarrow_i, b, t')$
    \ElsIf{$t$ and $t'$ are records}
        \State\Return$\solveRecordRecord(\hat{\ssat}_i, \vec{a}^\uparrow_i, \vec{a}^\downarrow_i, b, t, t')$
    \ElsIf{$t$ and $t'$ are choices}
        \State\Return$\solveChoiceChoice(\hat{\ssat}_i, \vec{a}^\uparrow_i, \vec{a}^\downarrow_i, b, t, t')$
    \ElsIf{$t$ or $t'$ is a switch}
        \State\Return$\solveSwitch(\hat{\ssat}_i, \vec{a}^\uparrow_i, \vec{a}^\downarrow_i, b, t, t')$
    \Else
        \State\Return$(\hat{\ssat}_i \cup \{\lnot b\}, \vec{a}^\uparrow_i, \vec{a}^\downarrow_i)$
    \EndIf
\end{algorithmic}
\end{algorithm}

\begin{figure}[t]
\begin{algorithm}[H]
\caption{$\assertWellFormed(\ssat_i, b, t)$}
\label{alg:assert-well-formed}
\begin{algorithmic}[1]
	\If{$t$ is a record $\record{\el{l_1}{b_1}{t_1},\dots,\el{l_n}{b_n}{t_n}}{}$ or $\choice{\el{l_1}{b_1}{t_1},\dots,\el{l_n}{b_n}{t_n}}{}$}
		\State$\ssat_{i+1} \gets \ssat_i \bigcup\limits_{\forall 1 \leq i, j \leq n\colon l_i = l_j} \{b \to \lnot (b_i \wedge b_j)\}$
	\ElsIf{$t$ is a switch \switch{\el{}{b_1}{t_1},\dots,\el{}{b_n}{t_n}}}
		\State$\ssat_{i+1} \gets \ssat_i \cup \{b_1 \lor \dots \lor b_n\} \bigcup\limits_{\forall 1 \leq i, j \leq n\colon i \neq j} \{b \to \lnot (b_i \wedge b_j)\}$
	\EndIf
	\State\Return$\ssat_{i+1}$
\end{algorithmic}
\end{algorithm}

\vspace{-4em}
\begin{algorithm}[H]
\caption{$\solveTupleTuple(\ssat_i, \vec{a}^\uparrow_i, \vec{a}^\downarrow_i, b, t, t')$}
\label{alg:solve-tuple-tuple}
\begin{algorithmic}[1]
    \State Let $t$ be of the form $\tuple{t_1\dots t_n}$
    \State $g \gets t'[\upvar{\vec{v}}/\vec{a}^\uparrow_i,\downvar{\vec{v}}/\vec{a}^\downarrow_i]$
    \If{$g = \tuple{t'_1 \dots t'_m}$ and $n  = m$}
        \For{$i\colon 1 \leq j \leq n$}
            \State$\ssat_{i+j}, \vec{a}^\uparrow_{i+j}, \vec{a}^\downarrow_{i+j} \gets \solve(\ssat_{i+j-1}, \vec{a}^\uparrow_{i+j-1}, \vec{a}^\downarrow_{i+j-1}, b, t_j, t'_j)$
        \EndFor
        \State\Return$(\ssat_{i+n}, \vec{a}^\uparrow_{i+n}, \vec{a}^\downarrow_{i+n})$
    \Else
        \State\Return$(\ssat_i \cup \{\lnot b\}, \vec{a}^\uparrow_{i+n}, \vec{a}^\downarrow_{i+n})$
    \EndIf
\end{algorithmic}
\end{algorithm}

\vspace{-4em}
\begin{algorithm}[H]
\label{alg:solve-nil-record}
\caption{$\solveNilRecord(\ssat_i, \vec{a}^\uparrow_i, \vec{a}^\downarrow_i, b, t')$}
\begin{algorithmic}[1]
    \State $g \gets t'[\upvar{\vec{v}}/\vec{a}^\uparrow_i,\downvar{\vec{v}}/\vec{a}^\downarrow_i]$
    \State Let $g$ be of the form $\record{\el{l'_1}{b'_1}{t'_1},\dots,\el{l'_m}{b'_m}{t'_m}}{}$
    \State\Return$(\ssat_i \cup \{b \to \bigwedge\limits_{j=1}^n \lnot b'_j\}, \vec{a}^\uparrow_{i}, \vec{a}^\downarrow_{i})$
\end{algorithmic}
\end{algorithm}

\vspace{-4em}
\begin{algorithm}[H]
\caption{$\solveRecordRecord(\ssat_i, \vec{a}^\uparrow_i, \vec{a}^\downarrow_i, b, t, t')$}
\label{alg:solve-record-record}
\begin{algorithmic}[1]
    \State $g \gets t'[\upvar{\vec{v}}/\vec{a}^\uparrow_i,\downvar{\vec{v}}/\vec{a}^\downarrow_i]$
    \If{$t = \record{\el{l_1}{b_1}{t_1},\dots,\el{l_n}{b_n}{t_n}}{}$}
        \For{$j\colon 1 \leq j \leq m$}
            \If{$\exists k\colon l_k \in t,\ l_k = l'_j$}
                \State$\ssat_{i+j}, \vec{a}^\uparrow_{i+j}, \vec{a}^\downarrow_{i+j} \gets \solve(\ssat_{i+j-1}, \vec{a}^\uparrow_{i+j-1}, \vec{a}^\downarrow_{i+j-1}, b \to b'_j \to b_k, t_k, t'_j)$
            \Else
                \State$\ssat_{i+j} \gets \ssat_{i+j-1} \cup \{b \to \lnot b'_j\}$
            \EndIf
        \EndFor
    \ElsIf{$t = \record{\el{l_1}{b_1}{t_1},\dots,\el{l_n}{b_n}{t_n}}{\downvar{v}}$}
        \For{$j\colon 1 \leq j \leq m$}
            \If{$\exists k\colon l_k \in t,\ l_k = l'_j$}
                \State$\ssat_{i+j}, \vec{a}^\uparrow_{i+j}, \vec{a}^\downarrow_{i+j} \gets \solve(\ssat_{i+j-1}, \vec{a}^\uparrow_{i+j-1}, \vec{a}^\downarrow_{i+j-1}, b \to b'_j \to b_k, t_k, t'_j)$
            \Else
                \State $\ssat_{i+j}, \vec{a}^\uparrow_{i+j}, \vec{a}^\downarrow_{i+j} \gets \text{set a new approximation for $\downvar{v}$ as $\record{\el{l'_j}{b'_j}{t'_j}}{}$ when $b$}$
            \EndIf
        \EndFor
    \EndIf
    \State\Return$(\ssat_{i+m}, \vec{a}^\uparrow_{i+m}, \vec{a}^\downarrow_{i+m})$
\end{algorithmic}
\end{algorithm}
\end{figure}

\begin{figure}[t]
\vspace{-8em}
\begin{algorithm}[H]
\caption{$\solveChoiceChoice(\ssat_i, \vec{a}^\uparrow_i, \vec{a}^\downarrow_i, b, t, t')$}
\label{alg:solve-choice-choice}
\begin{algorithmic}[1]
    \State $g \gets t[\upvar{\vec{v}}/\vec{a}^\uparrow_i,\downvar{\vec{v}}/\vec{a}^\downarrow_i]$
    \If{$t' = \choice{\el{l_1}{b_1}{t_1},\dots,\el{l_n}{b_n}{t_n}}{}$}
        \For{$j\colon 1 \leq j \leq m$}
            \If{$\exists k\colon l_k \in t,\ l'_k = l_j$}
                \State$\ssat_{i+j}, \vec{a}^\uparrow_{i+j}, \vec{a}^\downarrow_{i+j} \gets \solve(\ssat_{i+j-1}, \vec{a}^\uparrow_{i+j-1}, \vec{a}^\downarrow_{i+j-1}, b \to b_j \to b'_k, t_j, t'_k)$
            \Else
                \State$\ssat_{i+j} \gets \ssat_{i+j-1} \cup \{b \to \lnot b'_j\}$
            \EndIf
        \EndFor
    \ElsIf{$t' = \choice{\el{l_1}{b_1}{t_1},\dots,\el{l_n}{b_n}{t_n}}{\upvar{v}}$}
        \For{$j\colon 1 \leq j \leq m$}
            \If{$\exists k\colon l_k \in t,\ l_k = l'_j$}
                \State$\ssat_{i+j}, \vec{a}^\uparrow_{i+j}, \vec{a}^\downarrow_{i+j} \gets \solve(\ssat_{i+j-1}, \vec{a}^\uparrow_{i+j-1}, \vec{a}^\downarrow_{i+j-1}, b \to b_j \to b'_k, t_j, t'_k)$
            \Else
                \State $\ssat_{i+j}, \vec{a}^\uparrow_{i+j}, \vec{a}^\downarrow_{i+j} \gets \text{set a new approximation for $\upvar{v}$ as $\choice{\el{l_j}{b_j}{t_j}}{}$ when $b$}$
            \EndIf
        \EndFor
    \EndIf
        \State\Return$(\ssat_{i+m}, \vec{a}^\uparrow_{i+m}, \vec{a}^\downarrow_{i+m})$
\end{algorithmic}
\end{algorithm}

\vspace{-4em}
\begin{algorithm}[H]
\caption{$\solveSwitch(\ssat_i, \vec{a}^\uparrow_i, \vec{a}^\downarrow_i, b, t, t')$}
\label{alg:solve-switch}
\begin{algorithmic}[1]
    \If{$t = \switch{\el{}{b_1}{t_1},\dots,\el{}{b_n}{t_n}}$}
        \For{$j\colon 1 \leq i \leq n$}
            \State$\ssat_{i+j}, \vec{a}^\uparrow_{i+j}, \vec{a}^\downarrow_{i+j} \gets \solve(\ssat_{i+j-1}, \vec{a}^\uparrow_{i+j-1}, \vec{a}^\downarrow_{i+j-1}, b, t_j, t')$
        \EndFor
    \ElsIf{$t' = \switch{\el{}{b'_1}{t'_1},\dots,\el{}{b'_n}{t'_n}}$}
        \For{$i\colon 1 \leq j \leq n$}
            \State$\ssat_{i+j}, \vec{a}^\uparrow_{i+j}, \vec{a}^\downarrow_{i+j} \gets \solve(\ssat_{i+j-1}, \vec{a}^\uparrow_{i+j-1}, \vec{a}^\downarrow_{i+j-1}, b, t, t'_j)$
        \EndFor
    \EndIf
    \State\Return$(\ssat_{i+n}, \vec{a}^\uparrow_{i+n}, \vec{a}^\downarrow_{i+n})$
\end{algorithmic}
\end{algorithm}
\end{figure}

\end{document}